\documentclass[12pt]{article}

\usepackage{natbib}
\usepackage{hyperref}
\usepackage{graphicx}
\usepackage{amsmath, amssymb}
\usepackage{authblk} 
\usepackage{orcidlink} 
\usepackage{algorithm2e}
\usepackage{amsthm} 
\usepackage{booktabs}
\usepackage{tabularx}
\theoremstyle{remark}
\newtheorem{remark}{Remark}
\newtheorem{theorem}{Theorem}
\usepackage{arydshln}

\title{Optimal Cut-Point Estimation for Functional Digital Biomarkers: Application to Diabetes Risk Stratification via Continuous Glucose Monitoring}
\author[1,2]{Óscar Lado-Baleato}
\author[3]{Carla Díaz-Louzao}
\author[1,4]{Francisco Gude}
\author[5]{Marcos Matabuena}
\affil[1]{Research Methods Group (RESMET), Health Research Institute of Santiago de Compostela (IDIS), Spain}
\affil[2]{ISCIII Support Platforms for Clinical Research, Spain}
\affil[3]{Department of Mathematics, University of Coruña, Spain}
\affil[4]{Primary Care Centre, Concepción Arenal, Spain}
\affil[5]{Department of Biostatistics, Harvard University, USA}

\date{}

\begin{document}

\maketitle

\begin{center}
\end{center}


\maketitle
\abstract{Establishing optimal cut-offs for clinical biomarkers is a fundamental statistical problem in epidemiology, clinical trials, and drug discovery. While there is extensive literature regarding the definition of optimal cut-offs for scalar biomarkers, methodologies for analyzing random statistical objects in the more complex spaces associated with random functions and graphs - something increasingly required in the field of modern digital health applications - are lacking. This paper proposes a new, general, simple methodology for defining optimal cut-offs for random objects residing in separable Hilbert spaces. Its underlying motivation is the need to create new, digital health rules for the detection of diabetes mellitus, and thus better exploit the continuous high-dimensional functional information provided by continuous glucose monitors (CGM). A functional cut-off for identifying diabetes is offered, based on glucose distributional representations from CGM time series. This work may be a valuable resource for researchers interested in defining and validating new digital biomarkers for biosensor time series.}

\section{Introduction}
Technological advances in digital health have led to the appearance of small devices that can track different variables associated with human health. Modern, wearable devices and smartphones can now be used to track patients almost continuously, both in the hospital setting and at home. This has elevated healthcare into a new, continuous real-time dimension, improving disease management, diagnostics, and prevention, and potentially reducing healthcare costs across multiple settings \citep{dunn2018wearables}. 

Good examples of modern monitors include Holter devices \citep{del2005history}, which measure heart rate, and continuous glucose monitors (CGM) \citep{juvenile2008continuous}, which provide interstitial glucose concentrations every few minutes (indeed, such devices are revolutionizing continuous-time diabetes management). Physiological, environmental and biochemical data related to physical activity \citep{bunn2018current}, including values for electrocardiogram \citep{lee2018highly}, pulse oximetry and air-pollution variables \citep{pipek2021comparison}, etc., can also be tracked by similar devices. 

Naturally, this scenario requires the definition of new clinical biomarkers associated with high-frequency time series. These 'digital biomarkers' \citep{meister2016digital, babrak2019traditional} can be used to help define time-continuous endpoints. This, however, is not free of challenges in terms of data processing and statistical analysis. In this field, patients' time series are often summarized as statistical objects defined in functional, distributional and graph spaces. 

From a clinical perspective, digital biomarkers can provide a more accurate characterization of physiological states than can traditional diagnostic tests.  They introduce the temporal dimension at multiple points, offering a deeper understanding of metabolic and physiological profiles over extended periods. Traditionally, medical biomarkers are based only on periodic or sporadic blood tests under fasting conditions, on physical examinations, or on questionnaires. While these assess patient health at specific points in time under standardized conditions, allowing for easy comparison across groups, this ``snapshot'' approach provides limited information.  In digital health, biomarkers are derived from time series under free-living conditions, providing far more information - but making data analysis more difficult. Different authors \citep{matabuena2023distributional, ghosal2024multivariate} have suggested using distributional representations of time series as digital biomarkers, an approach that has improved the prediction of clinical outcomes in diabetes when followed using CGM.  Digital biomarkers are increasingly employed in areas such as cardiovascular disease \citep{chang2021deepheart}, gastroenterology \citep{dua2018monitoring}, mental health \citep{gentili2017longitudinal}, and, as mentioned, diabetes \citep{poolsup2013systematic}. New digital biomarkers will likely be defined for many other diseases, with patient summary representations taking different functional forms, e.g., distributional representations.

The present work proposes a simple method for estimating functional cut-offs for digital health biomarkers in Hilbert spaces, and demonstrates its potential for detecting diabetes mellitus when using CGM data collected over extended periods. This approach contrasts with the traditional diagnosis of diabetes, which relies on single-time-point measurements of fasting plasma glucose (FPG) and glycated haemoglobin (HbA1c). In a disease as heterogeneous as diabetes mellitus, in which different phenotypic characteristics complicate risk stratification and classification, this method could improve diagnosis. 

The remainder of the article is structured as follows. In the next subsection, the concept of CGM and its associated data analysis is introduced. Then, we present the main characteristics of optimal cut-off point estimation and how such cut-offs can be identified for functional data. Next section, introduces the proposed mathematical model; this is divided into five subsections. First, the distributional representation of CGM data (glucodensity) is discussed.  The problem of optimal cut-off estimation in separable Hilbert spaces is then formulated. The following subsection describes the practical calculation of the optimal cut-off for functional data. A bootstrap mechanism for estimating the functional cut-off is then introduced. Finally, this section concludes with details on software implementation. Then, in a following section, the proposed methodology is tested using simulated data. Next section addresses the problem of predicting the prevalence and incidence of diabetes based on functional representations from CGM. Lastly, the paper concludes with a broad discussion.

\subsection{Continuous glucose monitoring data analysis}
Diabetes mellitus is a chronic complex metabolic disorder affecting $830$ million people worldwide. Its prevalence is expected to rise over time by more than $10-$percent \citep{who_diabetes}. Diagnosing, risk stratifying, and managing the disease requires regular monitoring of blood glucose levels \citep{galicia2020pathophysiology}. Despite the strong evidence supporting the benefits of incorporating CGM into diabetes management, the use of these devices in healthy individuals remains a matter of debate. Some sources highlight this area as an ``evidence-free zone'' \citep{guess2023growing} suggesting caution in its use, while others outright advise against its use in healthy populations \citep{oganesova2024innovative}. However, CGM data could be valuable for identifying healthy individuals at higher risk of developing diabetes \citep{marco2022time,matabuena2024glucodensity}, offering the possibility of earlier intervention. 

The clinical interpretation of CGM data is challenging given the high-dimension nature of time series and the lack of temporal synchronization across patients. For each patient, the data reveals fluctuating values with peaks following meals, and troughs after fasting or physical activity. The timing of these peaks and troughs is influenced by a subject's behavior, and is therefore unpredictable, resulting in random patterns over the function's domain (see left plot of Figure \ref{fig0}). As a result, comparing raw CGM data across individuals (or even within the same individual) is unfeasible. Physicians therefore currently rely on summary measurements, such as the mean and standard deviation of a whole CGM time series. Several glycemic variability indices may also be used, including the area under the CGM curve (AUC), the Continuous Overall Net Glycemia Action (CONGA) index, the Low Blood Glucose Index (LBGI), the High Blood Glucose Index (HBGI), the Glycemic Risk Assessment Diabetes Equation (GRADE), and the Mean Amplitude of Glycemic Excursions (MAGE) metric, among others \citep{nguyen2020review}. The proportion of time spent below range ($<70$ mg/dL), within range ($70-140$ mg/dL) and above range ($>140$ mg/dL) are also commonly used \citep{beck2019validation, battelino2019clinical}. However, these indices provide only a partial view of glucose homeostasis complexity, since individuals with similar values for these indices may still show differences in their raw CGM records given that these discard some distributional information and do not fully take into account glucose dynamics.

To overcome the limitations of CGM indices, Matabuena et al. (2021) \citep{matabuena2021glucodensities} introduced a distributional representation of CGM data taking into account a patient's 'glucodensity'. This methodology has gained significant attention \citep{cui2023investigating, gomez2024insulin, matabuena2024multilevel} since it incorporates more detailed information than traditional CGM biomarkers. In essence, a glucodensity is a density function that captures the proportion of time spent at each glucose level, providing a comprehensive summary of all the data captured by a CGM tracker. Each subject generates a glucodensity, which can be analyzed as functional data within a distributional framework. This biostatistical concept can be viewed as a generalization of "time in range", extending it to cover the entire spectrum of glucose levels. The clinical interpretability of glucodensities, combined with the statistical tools developed to handle such data \citep{matabuena2022kernel, matabuena2023distributional}, make it a promising framework for predicting various conditions based on digital biomarkers. Indeed, distributional glucodensity-based summary methods can be used with any biological time series. Figure \ref{fig0} presents raw CGM data alongside the corresponding glucodensities, and their distributional representation, for two subjects: Subject $68$, a healthy control, and Subject $496$, who has diabetes. While the raw data reveal minimal discernible differences between these individuals, the glucodensity data and their distributional representation effectively capture variations in both the mean glucose value and its variability.

\begin{figure}[!htb]
\centering
\includegraphics[width=\textwidth]{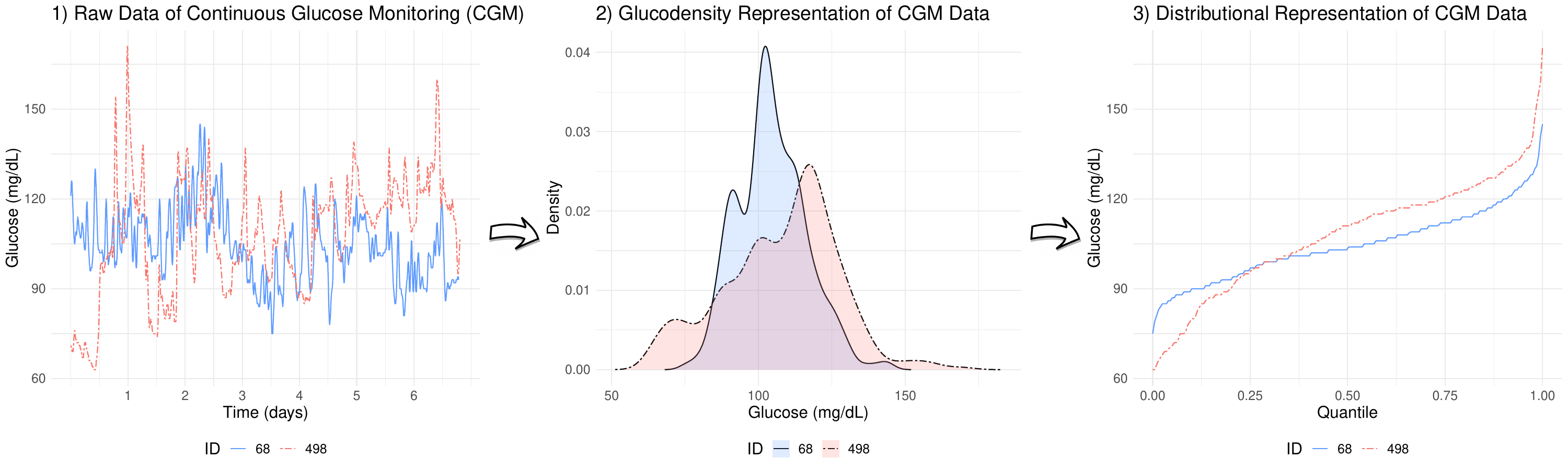} \\
\vspace{0.5cm}
\caption{Transformation from Continuous Glucose Monitoring (CGM) raw data into glucodensities and its distributional representation. This process is illustrated for two subjects; ID = $498$ who was diagnosed of diabetes during follow-up, and a healthy control ID = $68$.}
\label{fig0}
\end{figure}

\subsection{Optimal cut-points for scalar biomarkers and functional data analysis extensions}

The statistical field that focuses on defining optimal cut-offs aims to dichotomize continuous variables and thus create easy-to-implement clinical rules. One way of obtaining a cut-off is based on optimizing the ROC curve \citep{fawcett2006introduction}, such as maximizing the associated test's specificity or sensitivity. However, this method can be subjective since the choice of an optimal cut-off typically requires a balance that can be strongly influenced by the final application of that cut-off. Alternatively, the Youden Index \citep{fluss2005estimation, perkins2006inconsistency} maximizes the sum of both sensitivity and specificity to determine the optimal threshold. The diagnostic likelihood ratio \citep{boyko1994ruling} and cost-effectiveness analysis \citep{parsons2023integrating} further add to the statistical toolbox for optimal cut-off estimation. An updated review of the methods for defining optimal cut-offs can be found in \citep{lopez2014optimalcutpoints}.

In the field of digital health, the corresponding biomarkers are functions over time, referred to as functional data \citep{crainiceanu2024functional}.
Categorizing functional data biomarkers can, however, run into technical difficulties.  For example, techniques for multivariate analysis, such as logistic regression and discriminant analysis, are not directly applicable to functional data settings given the infinite-dimensional nature of the underlying random objects. A common solution is to propose specific functional data versions of such algorithms, leading to adaptations such as functional logistic regression \citep{araki2009functional} or functional discriminant analysis \citep{james2001functional, preda2007pls, shin2008extension}. Additional algorithms include the use of kernel non-parametric methods \citep{ferraty2003curves}, support vector machines \citep{rossi2006support}, principal components analysis \citep{hall2001functional, glendinning2003shape}, and functional depth classifiers \citep{cuevas2007robust}, among others. In theoretical work, it has been shown that perfect classification is feasible for Gaussian functional processes \citep{delaigle2012achieving, cuesta2023perfect} under certain conditions. However, in biomedical settings, classification challenges arise due to overlap between groups, making perfect classifications difficult. Such is the case with diabetes, a heterogeneous, multifactorial disease with different definitions according to the medical guidelines consulted.

The following lines introduce a simple strategy for defining the optimal cut-off for random objects in separable Hilbert spaces, exploiting previous work on optimal cut-offs for scalar biomarkers. 

\section{Mathematical models}
The objective of this section is to introduce a comprehensive framework for determining optimal cut-off points within statistical objects that take values in arbitrary separable Hilbert spaces, denoted as $\mathcal{H}$. Given our ultimate aim of applying such models to distributional CGM representations, we begin by defining the necessary mathematical representations. Subsequently, we present a generalized formulation of an optimal-cut algorithm designed for application in any separable Hilbert space. Then, to exemplify the algorithm's application in the context of the CGM scientific problem, we introduce an algorithm specific to this situation. Following this, we present an algorithm for quantifying the uncertainty in cut-point estimation using a bootstrap approach. Finally, we provide details on the practical implementation of the algorithm.

\subsection{Definition of distributional representation of glucose values: glucodensity}

We start with the formal definition of the distributional representation of glucose values. For a given patient \( i \), we denote the glucose monitoring data by pairs \( (t_{ij}, X_{ij}) \), where \( j = 1,\ldots,m_i \). Here, \( t_{ij} \) represents the recording times, which are typically equally spaced across the observation interval, and \( X_{ij} \in [a,b] \) is the glucose level at time \( t_{ij} \in [0, T_i] \). The values \( a \) and \( b \) represent the minimum and maximum range of the CGM monitor, which in our case are \( a = 40 \) mg/dL and \( b = 400 \) mg/dL. The number of records \( m_i \), the spacing between them, and the overall observation length \( T_i \) are allowed to vary across patients.

These data can be viewed as discrete observations of a continuous latent process \( Y_i(t) \), with \( X_{ij} = Y_i(t_{ij}) \). The glucodensity is defined as \( f_i(x) = F_i'(x) \), where \( F_i'(x) \) denotes the derivative of the cumulative distribution function \( F_i(x) \). This distribution function is given by

\begin{equation*}
F_i(x) = \frac{1}{T_i} \int_{0}^{T_i} \mathbb{I}(Y_i(t) \leq x) \, dt
\end{equation*}

\noindent for \( \underset{t \in [0, T_i]}{\inf} Y_i(t) \leq x \leq \underset{t \in [0, T_i]}{\sup} Y_i(t) \).

This expression represents the proportion of the observational interval during which the glucose level remains below the value \( x \). The functions \( F_i \) are a set of increasing functions from 0 to 1. For \( x \in [a, b] \), \( f_i(x) \) roughly measures the proportion of time that an individual patient spends at a continuous glucose concentration of \( x \) over the domain \( [a, b] \).

Density function modeling is a natural extension of traditional CGM data analysis methods called time-in-range metrics. These metrics measure the proportion of time individuals spend within certain glucose ranges. For example, hypoglycemia can be measured as \( \widehat{\text{Hypo}}_i = \frac{1}{n_i} \sum_{j=1}^{n_i} \mathbf{I}\{X_{ij} \leq 54\} \), but from a functional perspective, density functions provide a more detailed representation.

From a technical perspective, using probability distribution or density functions of individual CGM time series introduces challenges in statistical analysis because such representations are defined in a metric space without the linear structure of a vector space. In practice, to overcome these technical difficulties, it is common to consider the quantile function of the observed time series, \( \widehat{Q}_{i}(\rho) \), defined as

\[
\widehat{Q}_{i}(\rho) = \inf \left\{ t \in \mathbb{R} : \widehat{F}_{i}(t) = \frac{1}{n_i} \sum_{j=1}^{n_i} \mathbf{I}\{X_{ij} \leq t\} \geq \rho \right\}.
\]

Mathematically, in some modeling tasks, it is equivalent to consider the space of univariate probability distribution or density functions equipped with the 2-Wasserstein distance (see, for example, \citep{matabuena2021glucodensities, park2025fixedthresholdsoptimizingsummaries}).
\subsection{Defining optimal cut-points for random objects in separable Hilbert spaces}

In this section, we present a methodology to determine optimal cut-points for categorizing a biomarker \( Y \) that is defined within a separable Hilbert space \( \mathcal{H} \). For each subject \( i \), where \( i = 1, \dots, n \), let \( Y_i \in \mathcal{H} \) denote their biomarker measurement, and let \( Z_i \in \{0,1\} \) be a binary indicator of disease status—specifically, \( Z_i = 0 \) indicates the control (non-disease) group, and \( Z_i = 1 \) indicates the case (disease) group.

To classify subjects based on their biomarker profiles, we introduce a continuous parameter \( c \) within a specified range \( [l, u] \). This parameter helps partition the Hilbert space \( \mathcal{H} \) into two distinct regions:

\begin{itemize}
    \item \( \mathcal{A}^{c+} \): The region containing subjects predicted to have the disease.
    \item \( \mathcal{A}^{c-} \): The region containing subjects predicted not to have the disease.
\end{itemize}

We define an indicator function \( I_i^{c} \) for each subject \( i \) to denote their classification:

\[
I_i^{c} = \begin{cases}
1, & \text{if } Y_i \in \mathcal{A}^{c+}, \\
0, & \text{if } Y_i \in \mathcal{A}^{c-}.
\end{cases}
\]

To establish the partitioning of \( \mathcal{H} \) into \( \mathcal{A}^{c+} \) and \( \mathcal{A}^{c-} \), we define a threshold function \( h_{c}^\sigma(\rho) \) over the domain \( \mathcal{S} \) (the domain of \( Y \)):

\begin{equation}\label{eqn:def}
    h_{c}^\sigma(\rho) = \mu(\rho) + c\,\sigma(\rho), \quad \forall \rho \in \mathcal{S},
\end{equation}

\noindent where:
\begin{itemize}
    \item \( \mu(\rho)\) is the pointwise  centrality measure as  mean, median or mode of the biomarker at the point \( \rho \in  \mathcal{S}  \) and estimated with the data for a specific group or in overall.
    \item \( \sigma(\rho) \) is a scaling function that adjusts the threshold at each point \( \rho  \in \mathcal{S} \). 
\end{itemize}

Using \( h_{c}^\sigma(\rho) \), we define the regions:

\begin{align}
    \mathcal{A}^{c+} &= \left\{ f \in \mathcal{H} \,\big|\, f(\rho) \geq h_{c}^\sigma(\rho), \quad \forall \rho \in \mathcal{S} \right\}, \\
    \mathcal{A}^{c-} &= \left\{ f \in \mathcal{H}, f \notin  \mathcal{A}^{c+}  \right\}
\end{align}

Therefore, the indicator function \( I_i^{c} \) can be explicitly expressed as:

\[
I_i^{c} = \mathbb{I}\left( Y_i(\rho) \geq h_{c}^\sigma(\rho), \quad \forall \rho \in \mathcal{S} \right),
\]

where \( \mathbb{I}(\cdot) \) is the indicator function, returning 1 if the condition is true and 0 otherwise.

Our objective is to find the optimal value of \( c \in [l, u] \) that maximizes diagnostic performance. Specifically, we aim to maximize one of the following metrics:

\begin{align}
    \text{Sensitivity} &: \quad \max_{c \in [l, u]} \mathbb{P}(I_i^{c} = 1 \mid Z_i = 1), \\
    \text{Specificity} &: \quad \max_{c \in [l, u]} \mathbb{P}(I_i^{c} = 0 \mid Z_i = 0), \\
    \text{Youden Index} &: \quad \max_{c \in [l, u]} \left( \mathbb{P}(I_i^{c} = 1 \mid Z_i = 1) + \mathbb{P}(I_i^{c} = 0 \mid Z_i = 0) - 1 \right).
\end{align}

\begin{remark}
In Equation~\ref{eqn:def}, we define the partition function $h_{\rho}^{\sigma}(\cdot)$ in terms of a univariate parameter $c$. We focus on quantile functions, which are monotone increasing, and use the scalar $c$ to discriminate between disease and non-disease groups maybe it is enought. This works because the quantile functions are ordered for all probabilities $\rho \in [0,1]$.

However, in many applications involving more general data structures in a Hilbert space $\mathcal{H}$, discrimination between disease and non-disease groups can be more complex. In such scenarios, one may replace the scalar parameter $c$ by a function 
\[
c(\rho) \;=\; \sum_{j=1}^{m} c_j \,\phi_{j}(\rho),
\]
where $(c_1,\dots,c_m)\in \mathbb{R}^{m}$, and $\{\phi_{j}(\cdot)\}_{j=1}^m$ are predefined basis functions. The goal is then to select the optimal coefficients $(c_1,\dots,c_m)$ to achieve the best model discrimination. This functional extension offers greater flexibility and allows for more adaptive partitioning of $\mathcal{H}$, albeit at the cost of increased complexity.

In this paper, we focus on the scalar parameter $c$ because there is a clear biological boundary separating the diabetic and non-diabetic groups in our data and in terms of quantile functions. Nevertheless, the approach may be further generalized by reparameterizing $c$ as a function of the domain $p \in \mathcal{S}$ for modeling other clinical problems and data structures.
\end{remark}

\subsection{Practical calculation of optimal cut-points for distributional representations of CGM time series data}

In this context, we consider \( \mathcal{H} \) as the space of increasing quantile functions defined over the interval \([40, 400]\) mg/dL, which corresponds to the measurement range of Continuous Glucose Monitoring (CGM) devices. For each patient \( i \), we represent their random response as \( Y_i(\rho) = \widehat{Q}_i(\rho) \), where \( \rho \in [0,1] \) denotes the quantile probabilities and \( \widehat{Q}_i(\rho) \) is the estimated quantile function of their glucose measurements.

Given an independent and identically distributed (i.i.d.) sample \( \mathcal{D}_{n} = \{(Y_{i}, Z_{i})\}_{i=1}^{n} \) drawn from a joint distribution \( P \), where \( Y_i \in \mathcal{H} \) and \( Z_{i} \in \{0,1\} \) indicates disease status, we aim to estimate the threshold function \( h_{c}^\sigma(\rho) \) for any \( c \) in a specified range \( [l, u] \).

We estimate \( h_{c}^\sigma(\rho) \) using the sample mean of the quantile functions:

\begin{equation}
    h_{c}^\sigma(\rho) = \frac{1}{n} \sum_{i=1}^{n} Y_i(\rho) + c,
\end{equation}

\noindent where the term \( \frac{1}{n} \sum_{i=1}^{n} Y_i(\rho) \) is the empirical estimate of the mean function of the quantile function of the both groups considered (disease and non-disease). For simplicity in the clinical application examined, we assume \( \sigma(\rho) = 1 \) for all \( \rho \in \mathcal{S} \).

Using this estimated threshold function, we classify each observation \( Y_i \) into regions \( \mathcal{A}^{c+} \) or \( \mathcal{A}^{c-} \) by defining the indicator function \( I^{c}_{i} \):

\[
I^{c}_{i} = \mathbb{I}\left( Y_i(\rho) \geq h_{c}^\sigma(\rho) \quad \text{for all } \rho \in [0,1] \right),
\]

\noindent where \( \mathbb{I}(\cdot) \) is the indicator function that returns if the patient's quantile function exceeds the threshold at all probability  levels of the quantile function and zero otherwise. 

The clasification rules established allow us to empirically estimate the sensitivity, specificity, and the Youden index for each value of the parameter \( c \):

\begin{align}
    \widehat{\text{Sensitivity}}_{c} &= \frac{\sum_{i=1}^{n} \mathbb{I}(I^{c}_{i} = 1 \land Z_{i} = 1)}{\sum_{i=1}^{n} \mathbb{I}(Z_{i} = 1)}, \\
    \widehat{\text{Specificity}}_{c} &= \frac{\sum_{i=1}^{n} \mathbb{I}(I^{c}_{i} = 0 \land Z_{i} = 0)}{\sum_{i=1}^{n} \mathbb{I}(Z_{i} = 0)}, \\
    \widehat{\text{Youden Index}}_{c} &= \widehat{\text{Sensitivity}}_{c} + \widehat{\text{Specificity}}_{c} - 1.
\end{align}

By evaluating these metrics over a range of threshold values \( \{c_{1}, c_{2}, \ldots, c_{m}\} \), we identify the optimal \( c \) that maximizes the desired diagnostic criterion based on the empirical estimates. This approach enables us to select the cut-point that offers the best trade-off between sensitivity and specificity for disease detection using CGM data.

\subsection{Bootstrapping for uncertainty quantification in optimal cut-off determination}

Utilising the dataset \( \mathcal{D}_{n} = \{(Y_{i}, Z_{i})\}_{i=1}^{n} \), we identify the optimal cut-point \( \widehat{c} \) by optimising sensitivity, specificity, and the Youden Index through a one-dimensional \( c \)-process involving a family of functions:

\begin{equation*}
    h^{\sigma}_{c}(\rho) = \mu(\rho) + c\,\sigma(\rho), \quad c \in [l,u].
\end{equation*}

Bootstrap resampling is conducted \( B \) times, generating samples \( \mathcal{D}^{b}_{n} = \{(Y_{i}, Z_{i})\}_{i \in \mathcal{I}_{b}} \) by drawing with replacement from the index set \( \mathcal{I} = \{1, \dots, n\} \). This approach yields a series of optimal cut-points \( \{\widehat{c}_{b}\}_{b=1}^{B} \). We employ empirical quantiles \( \left[ \alpha/2,\, 1 - \alpha/2 \right] \) to calculate confidence intervals, establishing confidence regions for these optimal cut-points. This method facilitates inferential analysis on the optimal cut-points within a functional space using bootstrapping.

\begin{algorithm}[H]
\SetAlgoLined
\DontPrintSemicolon
\KwIn{Original dataset \( \mathcal{D}_{n} = \{(Y_{i}, Z_{i})\}_{i=1}^{n} \), Number of bootstrap samples \( B \) (e.g., \( B=1000 \)), Significance level \( \alpha \) (e.g., \( \alpha = 0.05 \))}
\KwOut{Confidence interval \( [c_{\text{lower}}, c_{\text{upper}}] \) for the optimal cut-point \( \widehat{c} \)}

\Begin{
    \For{\( b = 1 \) \KwTo B}{
        \textbf{Bootstrap Resampling:}\;
        Generate a bootstrap sample \( \mathcal{D}_{n}^{b} = \{(Y_{i}^{*}, Z_{i}^{*})\}_{i=1}^{n} \) by sampling with replacement from \( \mathcal{D}_{n} \).\;
        
        \textbf{Estimate Mean Function:}\;
        Compute the empirical mean function for the bootstrap sample:
        \[
        \mu^{*}_{b}(\rho) = \frac{1}{n} \sum_{i=1}^{n} Y_{i}^{*}(\rho).
        \]
        
        \textbf{Threshold Function and Classification:}\;
        \For{each candidate cut-point \( c \) in a grid \( [l, u] \)}{
            Define the threshold function:
            \[
            h_{c}^{*}(\rho) = \mu^{*}_{b}(\rho) + c.
            \]
            
            Classify each observation in the bootstrap sample:
            \[
            I_{i}^{c, *} = \mathbb{I} \left( Y_{i}^{*}(\rho) \geq h_{c}^{*}(\rho) \quad \forall \rho \in [0,1] \right).
            \]
        }
        
        \textbf{Compute Performance Metrics:}\;
        \For{each candidate cut-point \( c \)}{
            Compute the bootstrap estimates of sensitivity and specificity:
            \begin{align*}
                \widehat{\text{Sensitivity}}_{c}^{*} &= \frac{\sum_{i=1}^{n} \mathbb{I} (I_{i}^{c, *} = 1 \land Z_{i}^{*} = 1)}{\sum_{i=1}^{n} \mathbb{I} (Z_{i}^{*} = 1)}, \\
                \widehat{\text{Specificity}}_{c}^{*} &= \frac{\sum_{i=1}^{n} \mathbb{I} (I_{i}^{c, *} = 0 \land Z_{i}^{*} = 0)}{\sum_{i=1}^{n} \mathbb{I} (Z_{i}^{*} = 0)}.
            \end{align*}
            
            Compute the bootstrap estimate of the Youden Index:
            \[
            \widehat{\text{Youden Index}}_{c}^{*} = \widehat{\text{Sensitivity}}_{c}^{*} + \widehat{\text{Specificity}}_{c}^{*} - 1.
            \]
        }
        
        \textbf{Select Optimal Cut-Point for Bootstrap Sample:}\;
        Determine the optimal cut-point \( \widehat{c}_{b} \) for the bootstrap sample by maximizing the chosen criterion (e.g., Youden Index):
        \[
        \widehat{c}_{b} = \arg\max_{c} \widehat{\text{Youden Index}}_{c}^{*}.
        \]
    }
    
    \textbf{Construct Confidence Interval:}\;
    After all bootstrap samples are processed, obtain the empirical distribution of the optimal cut-points \( \{\widehat{c}_{1}, \widehat{c}_{2}, \dots, \widehat{c}_{B}\} \).
    
    Compute the \( \alpha/2 \) and \( 1 - \alpha/2 \) empirical quantiles of the bootstrap optimal cut-points to obtain the confidence interval:
    \[
    c_{\text{lower}} = \text{Quantile}_{\alpha/2}\left( \{\widehat{c}_{b}\}_{b=1}^{B} \right), \quad
    c_{\text{upper}} = \text{Quantile}_{1 - \alpha/2}\left( \{\widehat{c}_{b}\}_{b=1}^{B} \right).
    \]
    
    \textbf{Result:}\;
    The confidence interval for the optimal cut-point \( \widehat{c} \) is given by \( [c_{\text{lower}}, c_{\text{upper}}] \).
}

\caption{Bootstrap algorithm for confidence bands of optimal cut-points for our distributional clinical case study}
\end{algorithm}

\begin{theorem}
    The asymptotic distribution of the estimated optimal cut-point \( \widehat{c} \), with respect to sensitivity, specificity, and the Youden Index criteria, follows a Gaussian distribution. Consequently, the application of the bootstrap resampling technique for inference is justified.
\end{theorem}

\begin{proof}
Assuming the functions \( \mu \) and \( \sigma \) are fixed for simplicity, the regions \( \mathcal{A}^{c+} \) have a finite Vapnik–Chervonenkis (VC) dimension. This property ensures that the empirical process associated with the \( M \)-estimator, which maximises sensitivity, specificity, and the Youden Index, satisfies the Donsker property (see Chapter Bootstrapping Empirical Processes \cite{kosorok2008introduction}). Therefore, employing a bootstrap algorithm with Efron multipliers to construct confidence bands around the estimator \( \widehat{c} \) achieves universal consistency due to the guaranteed asymptotic normality of the estimator. 

In cases where the functions \( \mu \) and \( \sigma \) are estimated from the sample data, a data splitting strategy can be utilised to simplify the analysis. By dividing the data into separate subsets for estimating \( \mu \) and \( \sigma \) and for computing the optimal cut-point, we preserve the bootstrap's reliability and consistency.
\end{proof}

\subsection{Computational details}

The implementation of this methodology involves the functionalities of two R packages for optimal cutpoint estimation within functional data. First, the \texttt{fda} \cite{fda} package is employed to calculate the functional mean, representing the average curve of data. Subsequently, we estimate the maximum value of $c$ at which a given curve is above the functional cutpoint. Then, the \texttt{OptimalCutpoints} package \citep{lopez2014optimalcutpoints} facilitates the core analysis. Specifically, this package calculates a chosen performance metric (e.g., maximization of Youden's index) over the distribution of the obtained $c$ values. This cutpoint corresponds to the value of $c$ which offer the best performance metric, signifying the ideal separation between classes within the functional data. Finally, in order to ensure that the estimated functional cutpoint is monotonically increasing we use the R package \texttt{scam} \citep{pya2024scam} to smooth the obtained curve. Specifically, the estimated curve is smoothed using monotonically increasing P-splines \citep{pya2015shape}.

\section{Simulation study}\label{sec3}

In this section, we assess the efficacy of our proposal using simulated functional distributional data. The data generating process is based on the following formulation. 

\begin{figure}[!htb]
	\centering
	\includegraphics[width=0.9\textwidth]{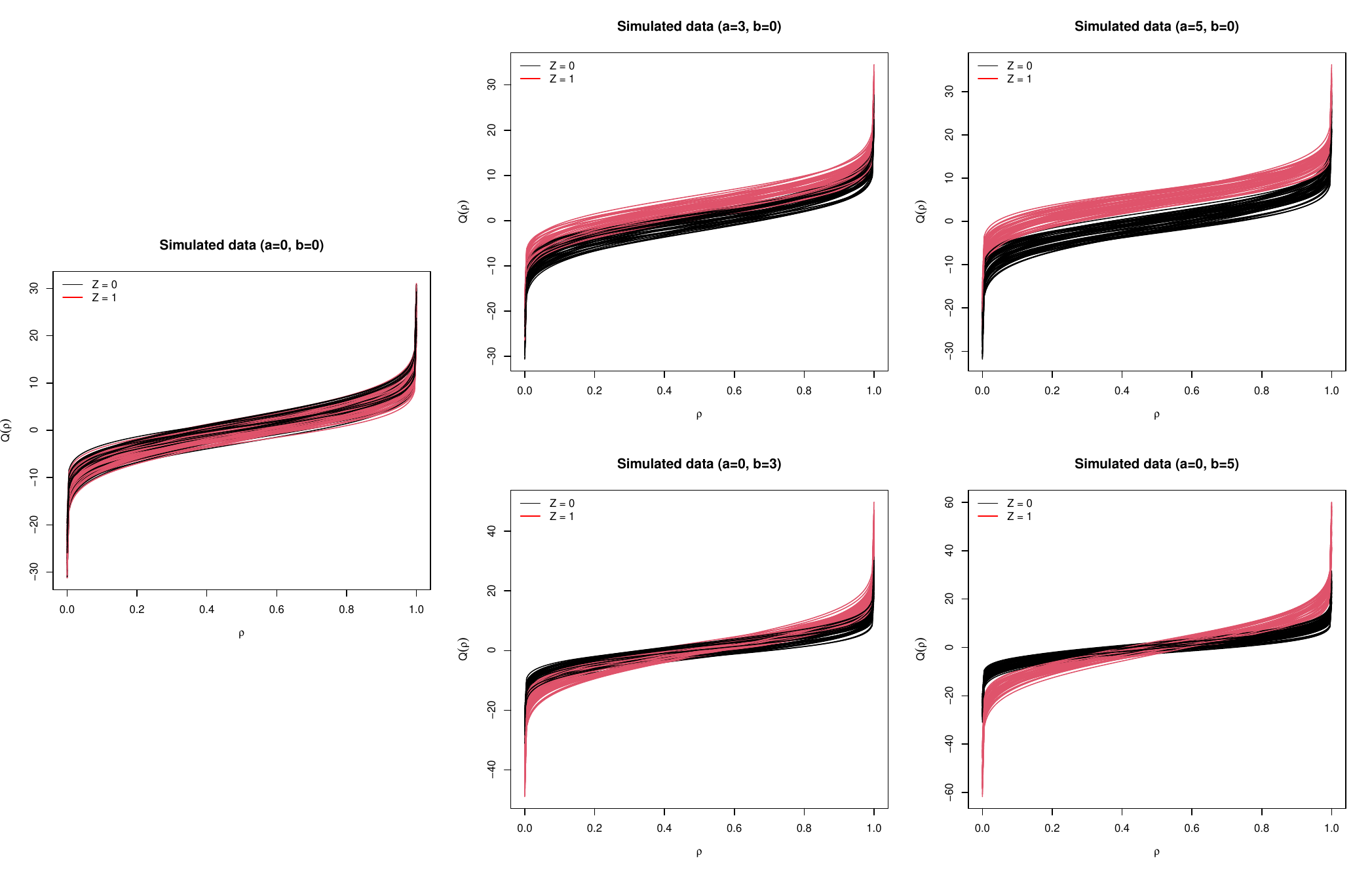} 
	\caption{Example of the simulated distributional data for different values of parameters $a$ and $b$, with $n=100$. }
	\label{fig1}
\end{figure}

Given a binary variable $Z \in Ber(0.5)$, and the quantile function $Q_0(\rho)$ from a truncate normal distribution $\text{TN}(\mu = 1, \sigma^2 = 1, a = -5, b = 5)$ the quantile function of  the functional distributional records are given by 

\begin{equation}
    Q(\rho) = aZ+U_1 + U_2 v + (5+b)Z U_3 Q_0(\rho) 
    \label{sim}
\end{equation}

\noindent where $U_1\sim U(-1, 1)$, $U_2\sim U(-1, 1)$ and $U_3\sim U(0.8, 1.2)$. The parameter $v$ was fixed to $v=2$ and it controls the degree of variability along $\rho$ of the simulated curves. The parameter $a$ governs the degree of difference between both groups regarding the location term of distributional functional data. Finally, $b$ is a parameter modulating the degree of difference between both groups due to the variability of the distributional functional data. A sketch of the simulated data is shown in Figure \ref{fig1} for different values of these parameters.

In this simulation scenario, we generated data according to equation \eqref{sim} using various values for parameters $a$ and $b$. Then, we assessed the discrimination capability of the proposed method between both groups. The disparities between the $Z=0$ and $Z=1$ groups are due to variations in either the location parameter (represented by varying $a$) or the scale parameter (represented by varying $b$). 

\begin{figure}[!htb]
	\centering
	\includegraphics[width=\textwidth]{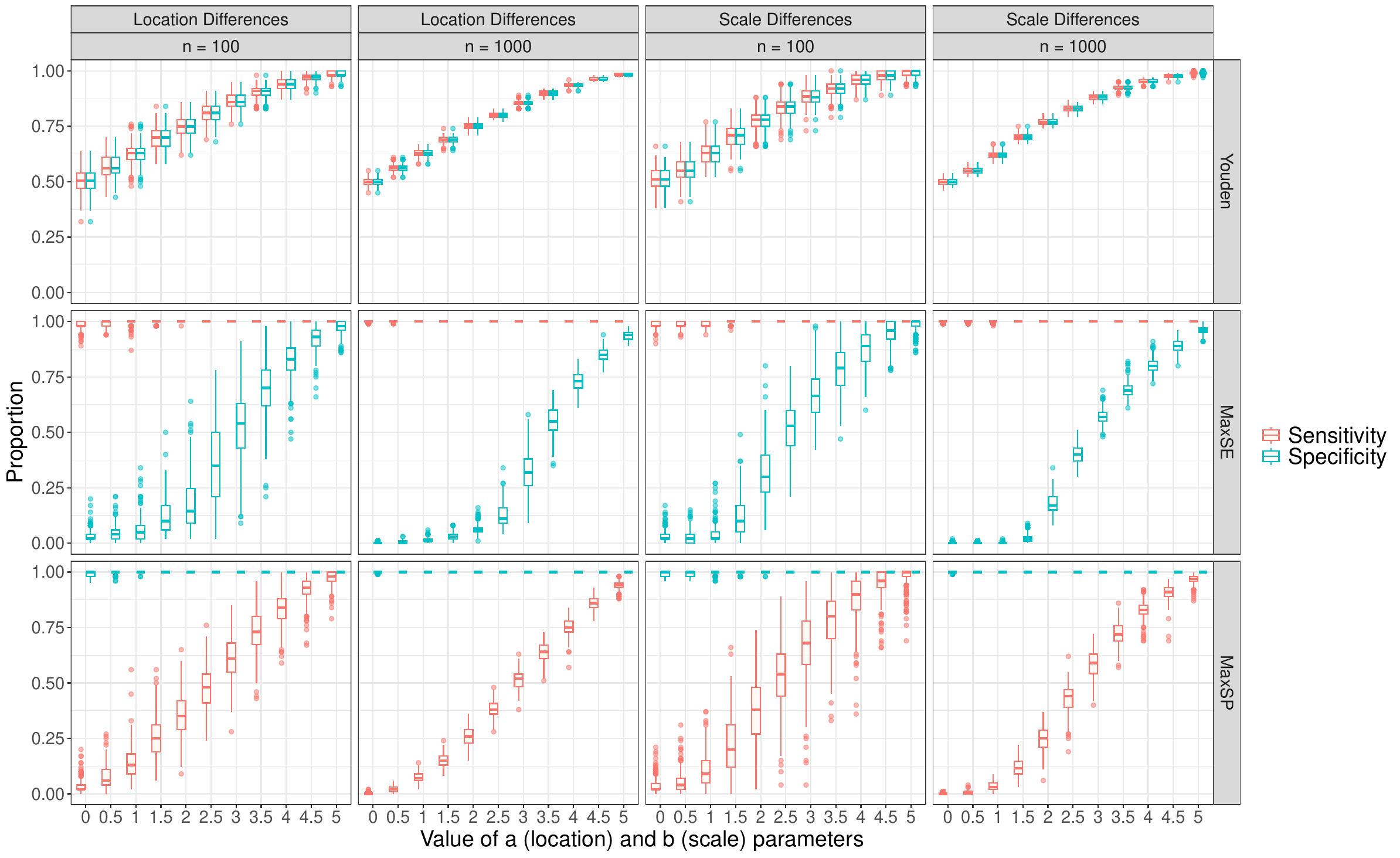} 
	\caption{Simulation results: estimated sensitivities and specificities for optimal cutoff curves based on Youden's, maximum specificity, and maximum sensitivity criteria across varying sample sizes and differences in location and scale parameters between groups.}
	\label{fig2}
\end{figure}

We estimated the functional cutpoint under three criteria: 1) Youden index criteria, 2) maximum sensitivity estimator , and 3) maximum specificity estimator. Sample sizes were set to $n=100$ and $n=1000$. The assessment was carried out across $R=1000$ replicates, with empirical sensitivity and specificity calculated for each classifier in every replicate.

In Figure \ref{fig2}, we present the sensitivity and specificity distribution across different values of $a$ and $b$. The Youden criteria display better performance improvements with increasing values of $a$ or $b$. Notably, the model exhibits the ability to accurately classify subjects regardless of whether the difference stems from location or scale parameters. As sample size increases the variability of the obtained sensitivities and specificities decreases. 

On  the other hand, the maximum sensitivity criteria consistently yield the highest sensitivity values regardless of the parameter values, while specificity gradually increases with greater disparities between groups. For values of the parameters $a$ and $b$ higher than $3.5$ we observe fair values for specificity. As in the Youden case, a higher sample size reduces the variability of the estimated specificities. Conversely, the maximum specificity criteria exhibit an inverse trend, showing a fair sensitivity for $a$ or $b$ values higher than $3.5$. Finally, the sensitivity variability decreases as sample size becomes higher. Both maximum sensitivity and specificity criteria works well with independence if the groups show differences on location or scale parameters. 

\section{Clinical validation of the optimal functional cutpoint}

In this section, optimal cut-off curves for identifying the prevalence and incidence of type $2$ diabetes are estimated using the distributional representation of CGM data. Establishing new diagnostic rules for type $2$ diabetes based on CGM data involves two key phases: first, identifying and classifying individuals using CGM data to detect current glycemic control abnormalities; and second, using the long-term follow-up data to predict future diabetes risk. Here, the prevalence of diabetes reflects the baseline disease status of subjects as diagnosed by primary care physicians. These diagnoses are based on the results of FPG and HbA1c tests. The incidence of diabetes here refers to the development of diabetes over a median follow-up period of $7.5$ years. This diagnosis is based on patients' electronic health records, which document repeated determinations of FPG and HbA1c over time, as well as documented diagnoses made by primary care physicians. Such longitudinal insights are unique to our cohort of subjects (see section below); other CGM studies often lack comparable long-term follow-up data. 

Based on this data we aim to establish new diagnostic rules for type $2$ diabetes based on CGM. This process involves two key phases: first, identifying and classifying individuals based on CGM data to detect current glycemic control abnormalities; and second, using the long-term follow-up data to predict future diabetes risk. The goal is to create more personalized and effective diagnostic criteria than traditional glycemic tests (i.e, fasting plasma glucose and glycated hemoglobin) or CGM indexes (e.g., mean glucose, standard deviation, mean amplitude glycemic excursions).

\subsection{AEGIS study}

The A-Estrada Glycation and Inflammation Study (AEGIS; trial NCT01796184) involved a unique population-based cohort, for which CGM data was collected along with long-term health outcomes, including the incidence of diabetes. Unlike most earlier studies, which primarily focused on patients with diabetes or which relied on short-term follow-ups, the AEGIS provides an extended longitudinal assessment of CGM data in both diabetic and non-diabetic individuals. This opens new avenues for assessing the association between CGM data and diabetes incidence risk in the general population. By capturing early glycemic patterns and their association with future diabetes risk, this study has the potential to redefine how CGM is used in disease prevention, beyond its traditional role in diabetes management.

Baseline data collection took place from $2012$ to $2015$, with participants undergoing structured assessments at their primary care centres, including answering lifestyle questionnaires, blood tests, and, for a subsample, a week-long CGM and dietary monitoring program (full details can be found in \citep{gude2017glycemic}). Health outcomes, including the incidence of type $2$ diabetes, were reviewed by August $2023$, providing a median follow-up of $7.5$ years (significantly longer than for most published CGM-based studies).

The AEGIS enrolled $1516$ subjects, of whom $622$ consented to participate in CGM. Among these, five participants did not comply with the study protocol and $37$ experienced biosensor disconnections during the monitoring period. Consequently, the final dataset comprised $580$ participants ($360$ women and $220$ men) who completed at least two days of CGM, including $65$ individuals diagnosed with diabetes at baseline (prevalence  $11.2\% $ $($CI95$\%$: $3.8\% - 14.1\%)$). 

At the start of each monitoring period, a research nurse inserted an Enlite$^{TM}$ CGM sensor (Medtronic, Inc., Northridge, CA, USA) subcutaneously into the subject's abdomen. The sensor continuously measured the interstitial glucose level ($40-400$ mg/dL), recording values every $5$ min. Participants were also provided with a conventional OneTouchR VerioR Pro glucometer (LifeScan, Milpitas, CA, USA) as well as compatible lancets and test strips for calibrating the CGM. All subjects were asked to make at least three capillary blood glucose measurements (usually before main meals) without checking the current CGM reading. The sensor was removed on the seventh day and the data downloaded and stored for later analysis. If the total data non-acquisition period (data   \textit{skips} ) per day exceeded $2$ h, all data for that day were discarded. 

Of the $515$ subjects with no diabetes at baseline, $26$ developed the disease by August $2023$ (incidence = $5.0\%$ (CI$95\%$: $3.4\% - 7.4\%$)). A participant was classified as having diabetes after meeting any of the following criteria:

\begin{enumerate}
\item Diagnosis of diabetes by a primary care doctor.
\item Prescription of antidiabetic drugs by a primary care doctor.
\item An HbA1c level of $>6.5\%$ and an FPG concentration of $126$ mg/dL or higher in a single blood test.
\item FPG or HbA1c values above the diagnostic cut-off points in two separate blood tests.
\end{enumerate}

\subsection{Optimal functional cutpoint estimation}\label{sec}

Figure \ref{figxx} shows the distributional representation of the CGM data for the AEGIS subjects (healthy subjects in blue, prevalent and incident cases of diabetes in red). The plot clearly shows distinct groupings, with subjects with diabetes having higher glucose levels than the rest. The separation of these two groups suggests that a functional cut-off could be established to effectively identify individuals with diabetes or who will develop it.

\begin{figure}[!htb]
\centering
\includegraphics[width=0.80\textwidth]{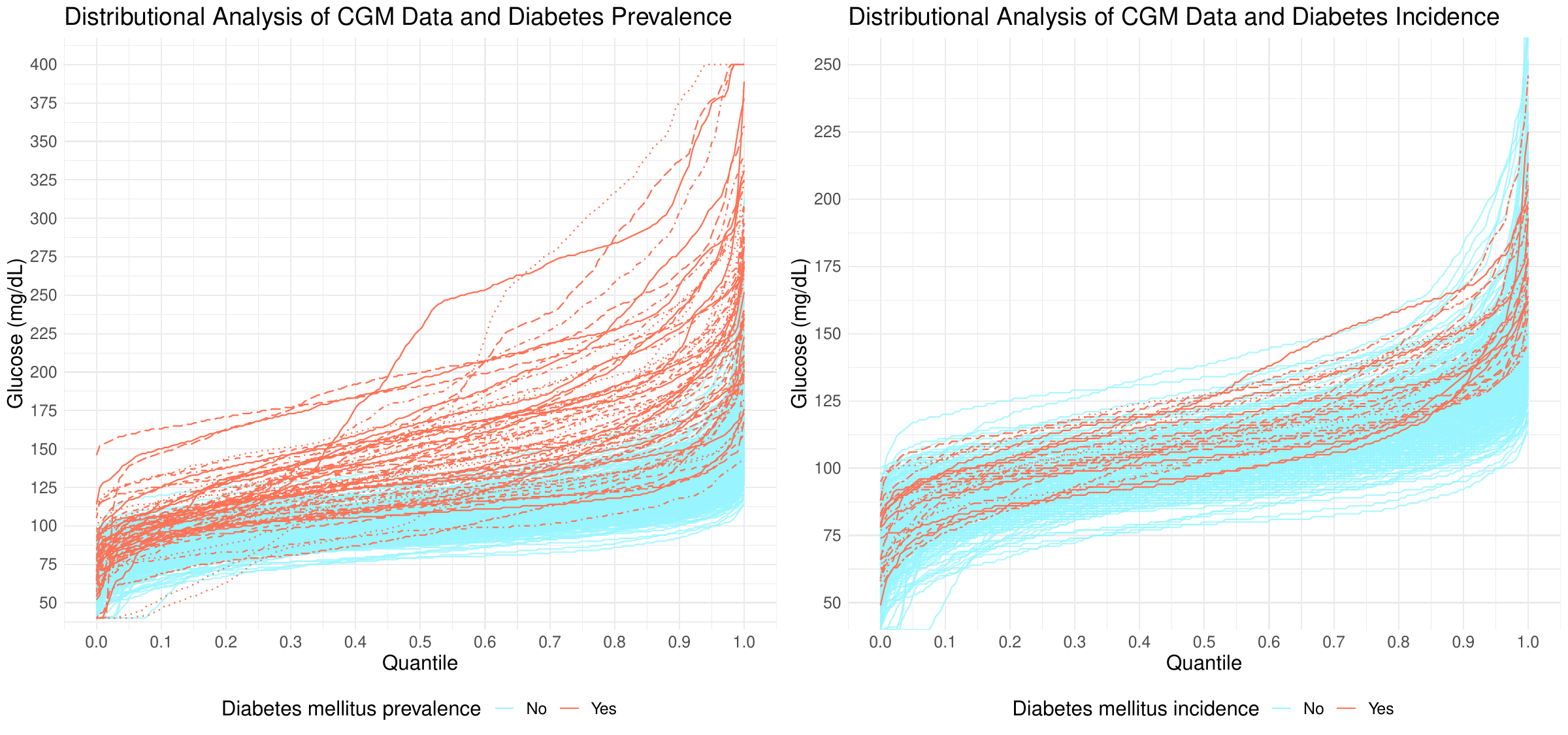}
\caption{Distributional representation of CGM data for our sample of subjects differentiating by baseline or follow-up diabetes status.}
\label{figxx}
\end{figure}

\begin{figure}[!htb]
\centering
	\includegraphics[width=0.85\textwidth]{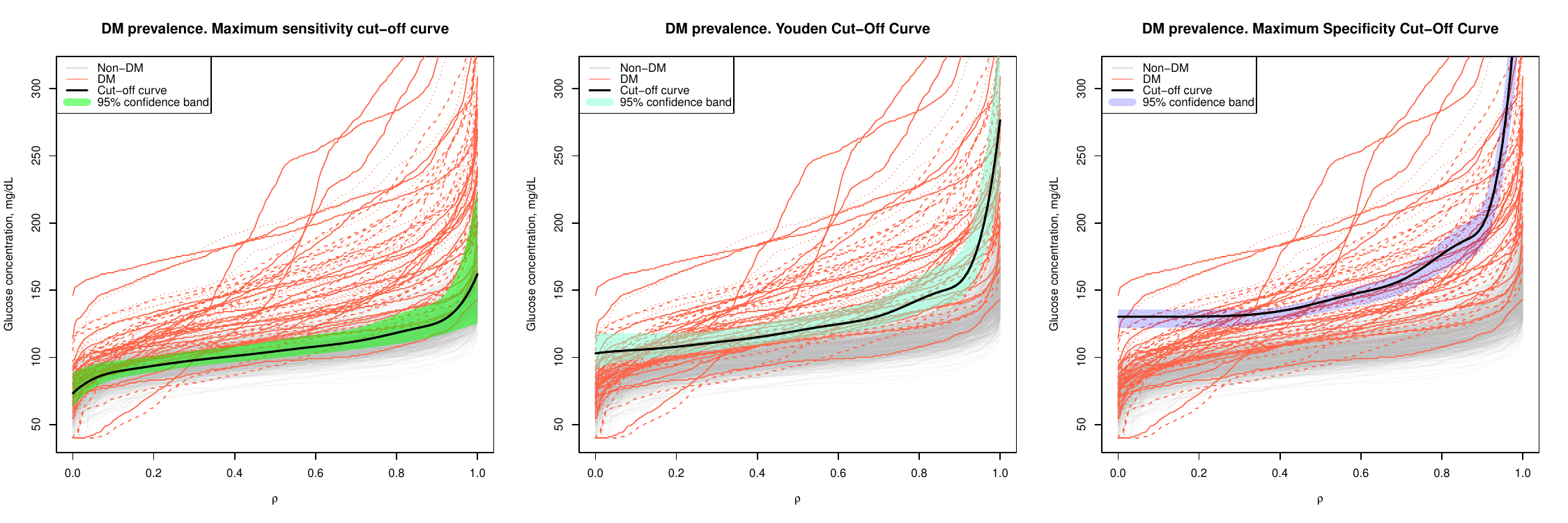}
		\includegraphics[width=0.85\textwidth]{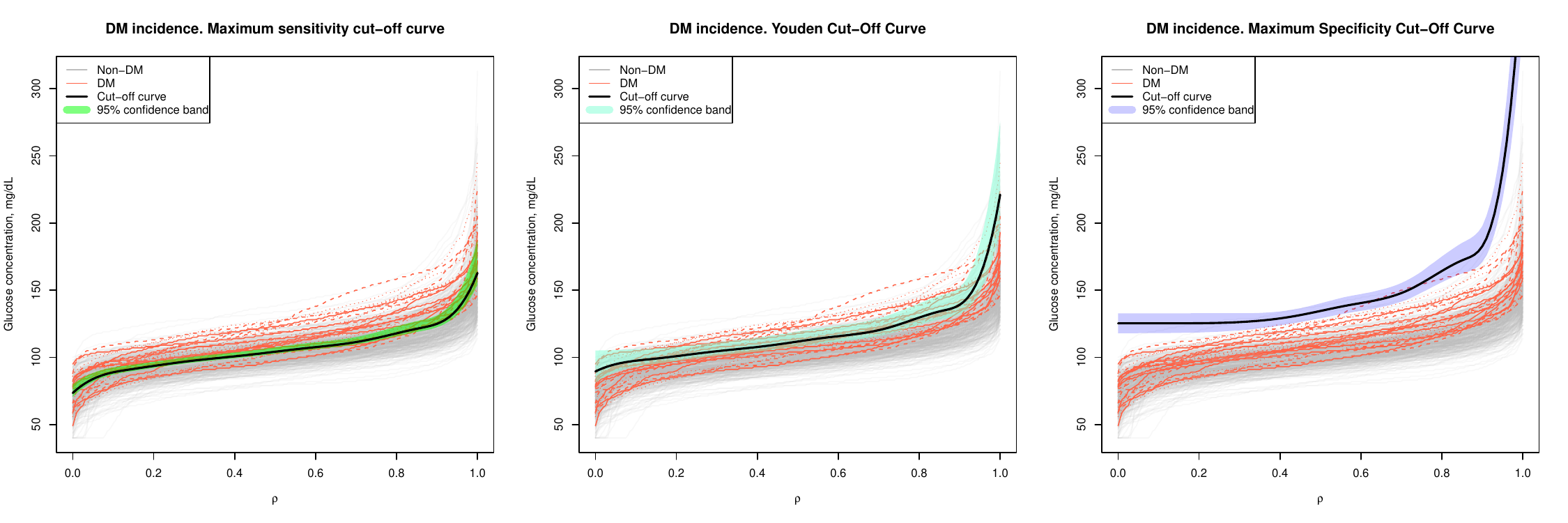}
	\caption{Optimal cut-off curves for predicting diabetes prevalence and incidence alongside with 95\% pointwise confidence band. Confidence bands were estimated based on the bootstrap scheme presented in algorithm 1.}
	\label{fig4}
\end{figure}

Figure \ref{fig4} shows the optimal cut-off curves for maximum sensitivity, the Youden Index, and maximum specificity, for predicting diabetes prevalence and incidence ($95\%$ pointwise confidence bands are shown alongside). For predicting diabetes prevalence, maximum sensitivity offers poor specificity  $0.10$ (CI$95\%: 0.06 - 0.80$). The median glucose value for this cut-off is $106.3$ mg/dL, with $\geq5.5\%$ of the time spent above the $140$ mg/dL threshold. The Youden Index shows a good balance between sensitivity and specificity, with a sensitivity of $0.89$ (CI$95\%: 0.78 - 0.98$) and a specificity of $0.92$ (CI$95\%: 0.79 - 0.99$). The median value of this cut-off curve rises to $119.8$ mg/dL, and the time above $140$ mg/dL to $\geq21.6\%$. Finally, the maximum specificity cut-off curve shows a low sensitivity of just $0.56$ (CI$95\%:0.42- 0.73$). The median value rises to $139.5$ mg/dL, and the curve is above $140$ mg/dL for $48.7\%$ of the monitoring time.  

When predicting diabetes incidence, maximum sensitivity shows poor specificity ($0.28$, 95\% CI: $0.16$ - $0.58$). The median glucose value for this cut-off curve is $106.5$ mg/dL, with $\geq4.9\%$ of the time spent above the $140$ mg/dL threshold. In contrast, the Youden Index provides a balanced trade-off between sensitivity and specificity, yielding a sensitivity of $0.84$ ($95$\% CI: $0.56$ - $0.96$) and a specificity of $0.69$ ($95$\% CI: $0.56$ - $0.93$). The median value of this cut-off curve increases until reaching $120.2$ mg/dL, with time spent above 140 mg/dL also increasing to $\geq20.5$\%. Finally, the optimal cut-off curve for maximum specificity shows very low sensitivity ($0.03$, $95$\% CI: $0.00$ - $0.13$). Its median value rises to $140.2$ mg/dL, remaining above the $140$ mg/dL threshold for $51.2\%$ of the time.

\subsection{Evaluating the discrimination capability of functional cutpoints}

The following discusses the estimated sensitivities and specificities for a range of values of the variable $c$ when predicting diabetes prevalence and incidence. The analysis begins with sensitivity (as shown in Figure \ref{figc}, depicted in red), which decreases as $c$ increases. This reduction in sensitivity occurs because higher values of $c$ make the test more strict, classifying fewer individuals as positive for diabetes. Conversely, specificity (shown in blue) improves with higher values of $c$, indicating that more individuals without diabetes are correctly classified as negative.  The shaded regions around the curves represent the 95\% pointwise confidence intervals, which capture the uncertainty in the estimated performance measures at each value of $c$. Importantly, the uncertainty is greater for the sensitivity data than for the specificity data, likely reflecting differences in the sample sizes of the patient subgroups. This disparity in uncertainty underscores the need to carefully account for variability when evaluating a functional cut-off, especially in the context of a heterogeneous disease like diabetes.

\begin{figure}[!ht]
\centering
	\includegraphics[width=\textwidth]{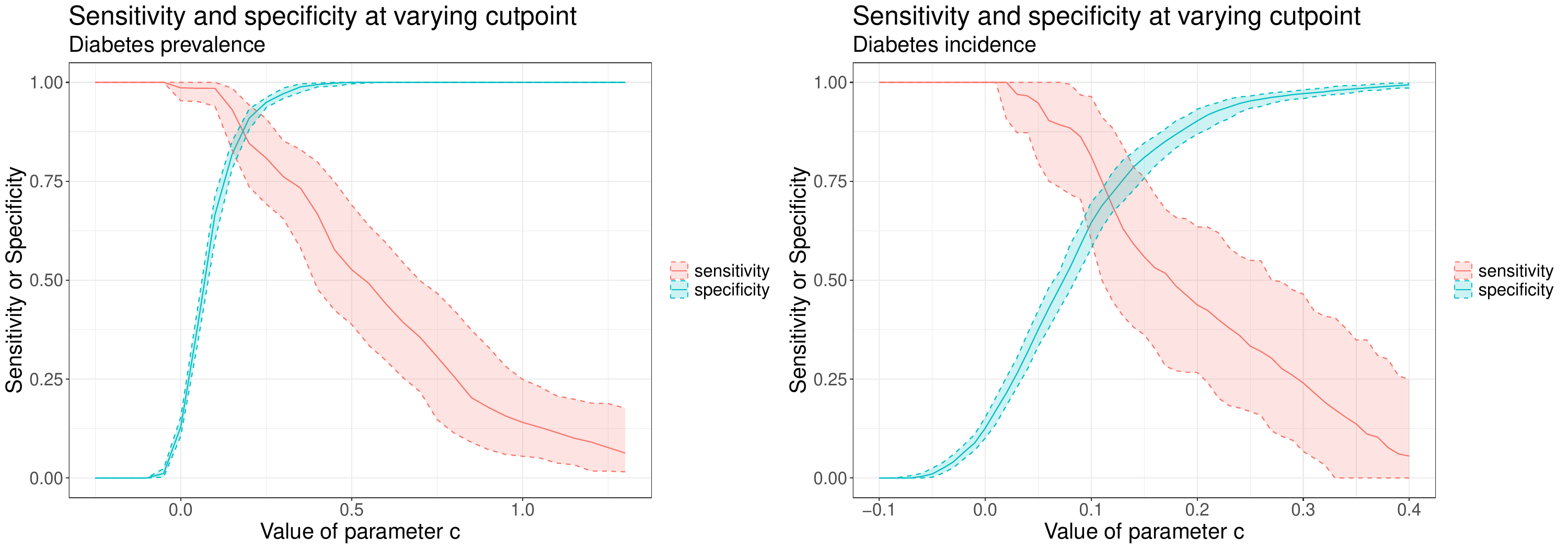}
	\caption{Sensitivity and specificity change across varying values of parameter $c$. We depict the values of the estimated metrics along with the $95\%$ pointwise confidence interval obtained based on the bootstrap procedure presented in algorithm 1.}
	\label{figc}
\end{figure}

By considering $specificities$ and $1 - sensitivities$ across the full range of $c$, the ROC curves shown in Figure \ref{fig5} are obtained. The estimated cut-off curves show better discrimination for diabetes prevalence than for diabetes incidence. While the prediction of diabetes prevalence shows an AUC of $0.952$ (CI$95\%: 0.918 - 0.976$), the prediction of future development of the disease shows a lower AUC of $0.791$  (CI$95\%: 0.698 - 0.864$).

\begin{figure}[!htb]
\centering
	\includegraphics[width=\textwidth]{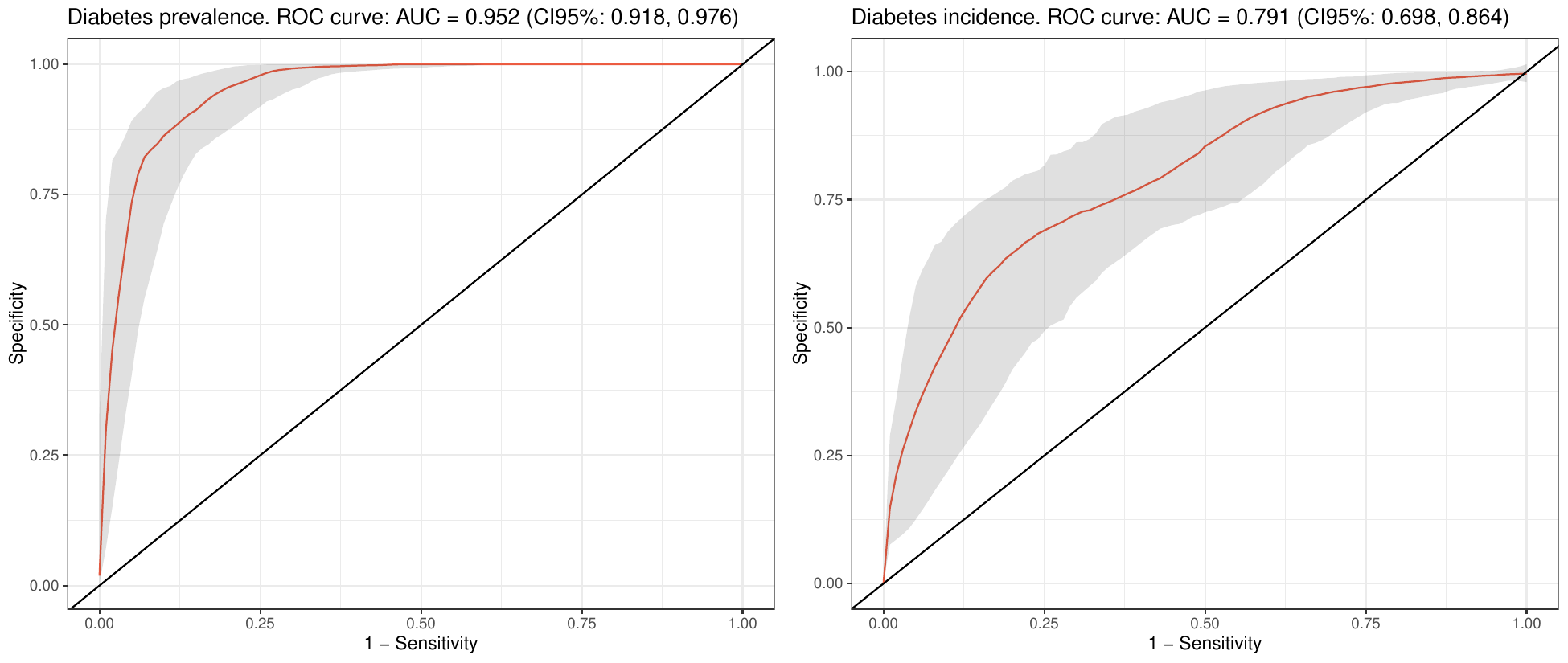}
	\caption{ROC curves, depicting specificities vs. ($1$ - sensitivities), obtained for a varying value of $c$ in the functional cut-point. Left plot shows the discrimination capability of the functional-cutpoint for diabetes prevalence, while right plot focus on predicting diabetes incidence. Confidence bands were estimated based on the bootstrap scheme presented in algorithm $1$.}
	\label{fig5}
\end{figure}

\subsection{Comparison with glycemic markers and CGM indexes}

This section compares the discriminatory capacity of the proposed distributional functional cut-off with various glycemic indices derived from CGM and the classic glycemic markers FPG and HbA1c. Table \ref{Tab_cutpoint} shows the sensitivity and specificity for the optimal cut-off as estimated under the Youden Index, maximum sensitivity and maximum specificity. For predicting diabetes prevalence, the functional cut-off shows a good, balanced performance under the Youden Index (sensitivity: $0.89$, specificity: $0.92$), comparable to the standard measures of HbA1c and FPG, which also demonstrate strong diagnostic capacity. 

The CGM indices (mean glucose [MG], standard deviation [SD], coefficient of variation [CV]  interquartile range [IQR], MAGE, CONGA, AUC and time above range [TAR]) show a poorer performance than the functional cut-off point under the Youden Index. While the maximum sensitivity values are perfect ($1.00$) across all methods, specificity is notably low for MG, IQR, MAGE, AUC and TAR, highlighting a trade-off when prioritizing sensitivity. Conversely, all methods reach perfect specificity under maximum specificity, but sensitivity varies, with HbA1c ($0.70$) and FPG ($0.57$) outperforming other metrics in maintaining balance. The functional cut-off sensitivity ($0.56$) is comparable to that for FPG. 

The results for diabetes incidence highlight the robust performance of the functional cut-off compared to the traditional variability indices FPG and HbA1c. Under the Youden Index, the functional cut-off achieves a sensitivity of $0.84$ and a specificity of $0.69$, balancing these metrics better than most other methods. If the Youden Index is calculated (sensitivity + specificity - 1), the functional cut-off obtains the best result ($0.53$). Notably, the functional cut-off outperforms MG (sensitivity: $0.88$, specificity: $0.60$) and SD (sensitivity: $0.84$, specificity: $0.54$), which both sacrifice specificity for sensitivity. For maximum sensitivity, all metrics achieve perfect sensitivity ($1.00$), but specificity is substantially higher for the functional cut-off ($0.28$) compared to the CV ($0.15$) and SD ($0.24$). This indicates that the functional cut-off maintains better overall discriminatory capacity, even when prioritizing sensitivity. For maximum specificity, the functional cut-off achieves perfect specificity ($1.00$), as do all other metrics, but sensitivity drops significantly ($0.03$). However, this trade-off is comparable to that seen for HbA1c (sensitivity: $0.04$) and MG (sensitivity: $0.00$).

\begin{table}[!ht]
\centering
\caption{Performance comparison of diagnostic metrics for diabetes prevalence and incidence using the Functional Cutpoint, glycemic variability indexes and glycemic markers. Sensitivity and specificity are evaluated under the Youden Index, maximum sensitivity, and maximum specificity criteria. For the functional cutpoint confidence intervals were estimated based on the bootstrap scheme presented in algorithm 1.}
\scalebox{0.60}{
\begin{tabular}{lcccccc}
 & \multicolumn{2}{c}{Youden Index} & \multicolumn{2}{c}{Maximum sensitivity} & \multicolumn{2}{c}{Maximum specificity} \\
 \cmidrule(lr){2-3} \cmidrule(lr){4-5} \cmidrule(lr){6-7}
 & Sensitivity & Specificity & Sensitivity & Specificity & Sensitivity & Specificity \\ \hline
\multicolumn{7}{l}{\textbf{DIABETES PREVALENCE}} \\
\cdashline{1-7}
\textbf{Functional Cutpoint} & \textbf{0.89 (0.78, 0.98)} & \textbf{0.92 (0.79, 0.99)} & \textbf{1.00 (1.00, 1.00)} & \textbf{0.10 (0.06, 0.80)} & \textbf{0.56 (0.42, 0.73)} & \textbf{1.00 (1.00, 1.00)} \\
MG & 0.77 (0.65, 0.86) & 0.94 (0.92, 0.96) & 1.00 (0.94, 1.00) & 0.00 (0.00, 0.01) & 0.48 (0.36, 0.60) & 1.00 (0.99, 1.00) \\
SD & 0.80 (0.68, 0.88) & 0.89 (0.86, 0.92) & 1.00 (0.94, 1.00) & 0.22 (0.19, 0.26) & 0.35 (0.24, 0.48) & 1.00 (0.99, 1.00) \\
CV & 0.81 (0.70, 0.89) & 0.73 (0.69, 0.77) & 1.00 (0.94, 1.00) & 0.12 (0.09, 0.15) & 0.17 (0.09, 0.28) & 1.00 (0.99, 1.00) \\
IQR & 0.72 (0.60, 0.82) & 0.93 (0.91, 0.95) & 1.00 (0.94, 1.00) & 0.00 (0.00, 0.01) & 0.34 (0.23, 0.46) & 1.00 (0.99, 1.00) \\
MAGE & 0.75 (0.63, 0.85) & 0.91 (0.89, 0.94) & 1.00 (0.94, 1.00) & 0.00 (0.00, 0.01) & 0.24 (0.14, 0.36) & 1.00 (0.99, 1.00) \\
CONGA & 0.73 (0.61, 0.83) & 0.92 (0.90, 0.94) & 1.00 (0.94, 1.00) & 0.24 (0.20, 0.28) & 0.16 (0.08, 0.28) & 1.00 (0.99, 1.00) \\
AUC & 0.82 (0.71, 0.90) & 0.93 (0.91, 0.95) & 1.00 (0.94, 1.00) & 0.00 (0.00, 0.00) & 0.48 (0.36, 0.60) & 1.00 (0.99, 1.00) \\
TAR ($\geq 180$ mg/dL) & 0.88 (0.78, 0.95) & 0.92 (0.90, 0.94) & 1.00 (0.94, 1.00) & 0.00 (0.00, 0.07) & 0.00 (0.00, 0.05) & 0.99 (0.98, 0.99) \\
TAR ($\geq 140$ mg/dL) & 0.85 (0.75, 0.92) & 0.88 (0.85, 0.91) & 1.00 (0.94, 1.00) & 0.00 (0.00,0.00) & 0.47 (0.35, 0.59) & 1.00 (0.99, 1.00) \\
FPG, mg/dL & 0.91 (0.82, 0.96) & 0.95 (0.93, 0.97) & 1.00 (0.94, 1.00) & 0.00 (0.00, 0.01) & 0.57 (0.44, 0.68) & 1.00 (0.99, 1.00) \\
HbA1c, \% & 0.92 (0.84, 0.97) & 0.94 (0.92, 0.96) & 1.00 (0.94, 1.00) & 0.37 (0.33, 0.42) & 0.70 (0.57, 0.80) & 1.00 (0.99, 1.00) \\ \hline
\multicolumn{7}{l}{\textbf{DIABETES INCIDENCE (7.5 YEARS)}} \\
\cdashline{1-7}
\textbf{Functional Cutpoint} & \textbf{0.84 (0.56, 0.96)} & \textbf{0.69 (0.56, 0.93)} & \textbf{1.00 (1.00, 1.00)} & \textbf{0.28 (0.16, 0.58)} & \textbf{0.03 (0.00, 0.13)} & \textbf{1.00 (0.99, 1.00)} \\
MG & 0.88 (0.68, 0.97) & 0.60 (0.56, 0.65) & 1.00 (0.86, 1.00) & 0.38 (0.33, 0.42) & 0.00 (0.00, 0.13) & 0.99 (0.98, 0.99) \\
SD & 0.84 (0.63, 0.95) & 0.54 (0.49, 0.58) & 1.00 (0.86, 1.00) & 0.24 (0.20, 0.28) & 0.00 (0.00, 0.13) & 0.99 (0.98, 0.99) \\
CV & 0.68 (0.46, 0.85) & 0.58 (0.54, 0.62) & 1.00 (0.86, 1.00) & 0.15 (0.12, 0.18) & 0.00 (0.00, 0.13) & 0.99 (0.98, 0.99) \\
IQR & 0.60 (0.38, 0.78) & 0.81 (0.77, 0.84) & 1.00 (0.86, 1.00) & 0.30 (0.26, 0.34) & 0.04 (0.00, 0.20) & 1.00 (0.99, 1.00) \\
MAGE & 0.96 (0.79, 0.99) & 0.43 (0.39, 0.48) & 1.00 (0.86, 1.00) & 0.39 (0.35, 0.44) & 0.00 (0.00, 0.13) & 0.99 (0.98, 0.99) \\
CONGA & 0.72 (0.50, 0.87) & 0.73 (0.69, 0.77) & 1.00 (0.86, 1.00) & 0.37 (0.33, 0.42) & 0.04 (0.00, 0.20) & 0.99 (0.98, 0.99) \\
AUC & 0.88 (0.68, 0.97) & 0.59 (0.55, 0.63) & 1.00 (0.86, 1.00) & 0.36 (0.32, 0.40) & 0.00 (0.00, 0.13) & 0.99 (0.98, 0.99) \\
TAR ($\geq 180$ mg/dL) & 0.40 (0.21, 0.61) & 0.85 (0.81, 0.88) & 1.00 (0.86, 1.00) & 0.00 (0.00, 0.00) & 0.00 (0.00, 0.13) & 0.99 (0.98, 0.99) \\
TAR ($\geq 140$ mg/dL) & 0.72 (0.50, 0.87) & 0.79 (0.75, 0.83) & 1.00 (0.86, 1.00) & 0.38 (0.34, 0.43) & 0.04 (0.00, 0.20) & 0.99 (0.98, 0.99) \\
FPG, mg/dL & 0.76 (0.54, 0.90) & 0.76 (0.72, 0.80) & 1.00 (0.86, 1.00) & 0.45 (0.41, 0.50) & 0.08 (0.00, 0.26) & 0.99 (0.98, 0.99) \\
HbA1c, \% & 0.72 (0.50, 0.87) & 0.80 (0.76, 0.83) & 1.00 (0.86, 1.00) & 0.14 (0.11, 0.18) & 0.04 (0.00, 0.20) & 0.99 (0.98, 0.99) \\ \hline
\end{tabular}}
\label{Tab_cutpoint}
\end{table}

Table \ref{Tab1} shows the AUC values for predicting diabetes prevalence and incidence. The proposed functional classifier consistently outperforms the other metrics, achieving a high AUC for diabetes prevalence ($0.952$) and the highest for diabetes incidence ($0.791$).

For diabetes prevalence, most of the glycemic indices show strong discrimination, with AUC values above $0.80$. The functional cut-off shows one of the highest discriminatory capacities (AUC = $0.952$), only surpassed by HbA1c (AUC = $0.979$). TAR ($>140$ mg/dL) and FPG show AUC values ($0.947$ and $0.945$ respectively), close to that of the functional cut-off. Other metrics, such the AUC ($0.937$) and CONGA ($0.905$), also perform well, while CV shows relatively weaker discrimination (AUC = $0.836$).

\begin{table}[!htb]
\caption{Discrimination capability of the functional cutpoint in comparison with the main glycemic indexes one in terms of AUC. For the functional cutpoint confidence intervals were obtained based on the bootstrap scheme presented in algorithm $1$.}
\centering
\begin{tabular}{lcc}
 & Diabetes Prevalence & Diabetes Incidence \\
 \hline
\textbf{Functional Cutpoint} & \textbf{0.952 (0.918, 0.976)} & \textbf{0.791 (0.698, 0.864)} \\
\cdashline{1-3}
MG & 0.885 (0.820, 0.950) & 0.734 (0.660, 0.807) \\
SD & 0.892 (0.845, 0.939) & 0.667 (0.578, 0.756) \\
CV & 0.836 (0.783, 0.889) & 0.598 (0.498, 0.698) \\
IQR & 0.859 (0.792, 0.925) & 0.689 (0.599, 0.779) \\
MAGE & 0.847 (0.781, 0.915) & 0.688 (0.611, 0.765) \\
CONGA & 0.905 (0.863, 0.946) & 0.713 (0.638, 0.788) \\
AUC & 0.937 (0.896, 0.979) & 0.716 (0.642, 0.790) \\
TAR (\textgreater{}180 mg/dL) & 0.924 (0.883, 0.964) & 0.545 (0.449, 0.641) \\
TAR (\textgreater{}140mg/dL) & 0.947 (0.921, 0.973) & 0.720 (0.647, 0.793)\\
\cdashline{1-3}
FPG & 0.945 (0.898, 0.992)  & 0.767 (0.704, 0.829) \\
HbA1c & 0.979 (0.961, 0.998) &  0.742 (0.652, 0.792) \\
 \hline
\end{tabular}
\label{Tab1}
\end{table}

For diabetes incidence, however, the discriminatory capacity of all metrics is generally lower, with reduced AUC values. The functional cut-off retains its lead, achieving the highest AUC for incidence ($0.791$), followed by FPG ($0.767$) and HbA1c ($0.742$). In general, the CGM-derived indices show AUCs below $0.75$ - far below the performance of the functional cut-off. This highlights the latters ability to identify individuals at higher risk of developing diabetes. In contrast, the metrics CV ($0.598$) and TAR ($>180$ mg/dL, $0.545$) show poor performance in predicting diabetes incidence, suggesting they may be less effective in capturing glycemic variability patterns associated with diabetes progression.

In conclusion, the proposed functional cut-off shows a robust and consistent performance in identifying patients with diabetes, as well as those at high risk of developing the disease in the future. The estimated AUC values are comparable to those of the best-performing CGM-derived indices and demonstrate superior predictive capacity, particularly for diabetes incidence. These findings suggest that the functional cut-off is a reliable and potentially superior metric for use in both clinical and research contexts, with significant potential for identifying individuals with diabetes as well as those at higher risk of developing the disease.

\subsection{Validation with an external cohort}

This section reports the validation of the estimated functional cut-off for predicting diabetes incidence in a new cohort of subjects. The proposed methodology was applied to a previously published dataset \citep{colas2019detrended} (provided as Supplementary Material). This study included $208$ subjects, all of whom were healthy at the start, with $17$ developing type $2$ diabetes by its end. Participants were selected from the Hypertension and Vascular Risk Outpatient Clinic at the University Hospital of Móstoles, Madrid, between January $2012$ and May $2015$. Inclusion criteria required participants to be aged between $18$ and $85$ years, with a prior diagnosis of essential hypertension and no history of diabetes mellitus or antidiabetic treatment.

At the start of the study, a CGM (iPro, MiniMed, Northridge, CA, USA) was used to record glucose levels for a minimum of $24$ h, with samples taken every $5$ min. Patients were subsequently followed-up every six months until either a diagnosis of type $2$ diabetes was made or the study ended.  Diagnoses were confirmed via FPG $\geq126$ mg/dL and/or HbA1c $\geq6.5$\%, with both criteria requiring confirmation in a second measurement. The median follow-up period was $33$ months (range: $6$ to $72$ months). During this period, $17$ new cases of diabetes were identified, with a median time to diagnosis of $33.8$ months (IQR: $24.1$ months).

\begin{figure}[!htb]
\centering
	\includegraphics[width=\textwidth]{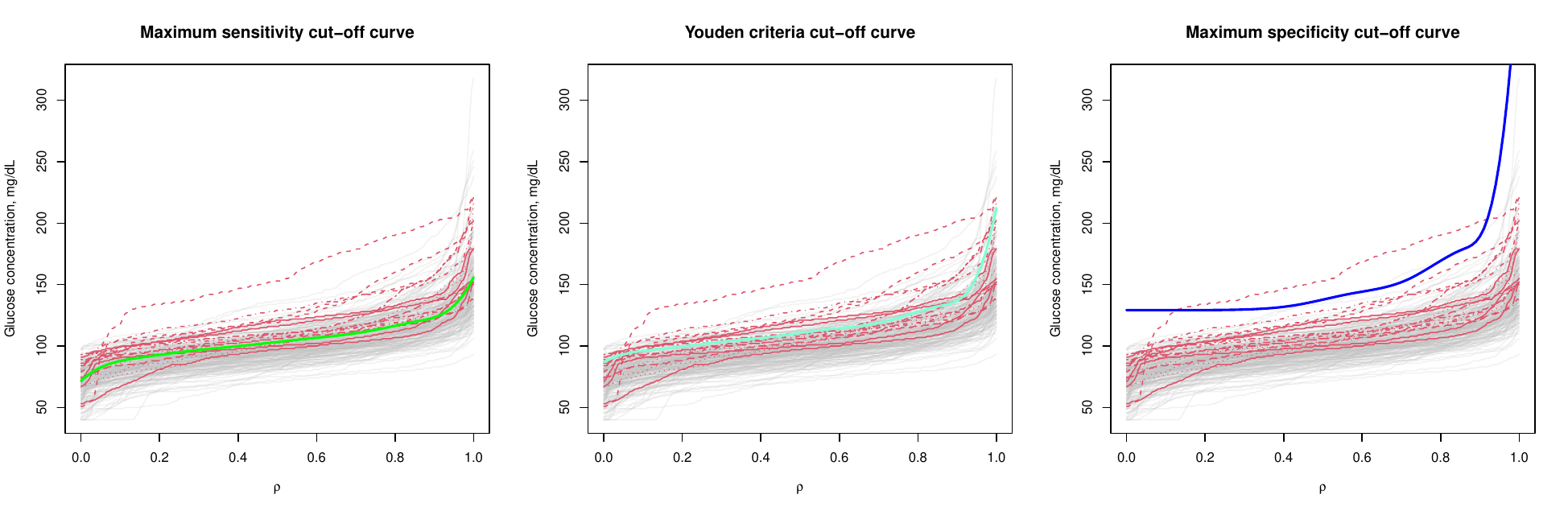}
	\caption{Estimated functional cut-off curves estimated using the AEGIS sample over the distributional representation of the CGM data collected by \citep{colas2019detrended}. Subjects who remained healthy at the end of the follow-up are depicted in grey, while those who developed diabetes are shown in red.}
	\label{fig8}
\end{figure}

Figure \ref{fig8} shows the distributional representation of the CGM data from \citep{colas2019detrended} alongside the estimated functional cut-off curves obtained as in Section \ref{sec}. The cut-off curves are shown on the same scale as the observed data. For maximum sensitivity, a sensitivity of $0.94$ ($95\%$ CI: $0.82$, $0.99$) and a specificity of $0.23$ ($95\%$ CI: $0.17$, $0.29$) was obtained. For the Youden Index, the estimated sensitivity and specificity was $0.76$ ($95\%$ CI: $0.56$, $0.96$) and $0.70$ ($95\%$ CI: $0.63$, $0.76$), respectively. Finally, maximum specificity identified only one case of diabetes incidence, achieving high specificity ($0.99$, $95\%$ CI: $0.98$, $1.00$) but low sensitivity ($0.05$, $95\%$ CI: $0.01$, $0.17$).

\section{Discussion}

This paper introduces a novel methodology for determining optimal cut-offs for emerging biomarkers in digital health applications, both for mathematical functions and probability distributions. The present approach leveraged high-resolution glucose time series data from CGM devices, with clinical outcome characterized using distributional quantile representations. The functional cut-off outperformed traditional CGM biomarkers such as the AUC, both in prevalence and incidence detection. To the best of our knowledge, this is the first instance in which a functional cut-off has been used to characterize diabetes mellitus using glucose time series data with a follow-up period involving nearly eight years. Ongoing CGM studies with healthy populations typically involve a reasonable number of subjects but have follow-up periods of less than four years. 

In the coming years, digital biomarkers are expected to play a crucial role in characterizing a range of metabolic disorders and diseases, as well as in assessing patient prognosis for conditions such as gestational diabetes, where traditional oral glucose tolerance tests are not always advisable. Recent studies have addressed biomarker assessment using ROC analysis within standard functional spaces, relying on restrictive linear and Gaussian process assumptions specifically designed for random functions in spaces defined by \( L^{2}([0,1])\) \citep{bianco2024roc}. However, such approaches are not suitable for distributional representations in digital health. The present work focuses on a more general non-parametric model applicable to any random object in a separable Hilbert space with a linear vector-valued structure or embedding, as in the case of probability distributions under the lens of $2$--Wasserstein metrics. The study specifically tackles the problem of optimal cut-off detection, which is essential for defining new digital biomarkers.

As digital health continues to evolve, the need to contemplate random objects defined in separable Hilbert spaces will likely become increasingly prevalent. The discussed methodology could be applied to a wide range of clinical problems across different domains (even medical imaging for neurodegenerative diseases), providing opportunities to take preventive measures and prescribe more informed treatments. Clinical trials too might benefit from its use \citep{beauchamp2020use}.


\section{Competing interests}
No competing interest is declared.

\section{Author contributions statement}
O.L.-B. was responsible for manuscript writing, performing the simulation study, and conducting data analysis. C.D.-L. contributed to the study design, manuscript revision, and result interpretation. M.M. developed the mathematical framework, wrote part of the manuscript, and contributed to the conceptual framework of the study. F.G. provided the data and supported its analysis.

\section{Acknowledgments}
This study has been funded by Instituto de Salud Carlos III (ISCIII) through the projects PI11/02219, PI16/01395, PI20/01069; by the European Union (Programme for Research and Innovation, 2021–2027; Horizon Europe, EIC Pathfinder: 101161509-GLUCOTYPES); and by the Network for Research on Chronicity, Primary Care, and Health Promotion, ISCIII, RD21/0016/0022; and co-funded by the European Union-NextGenerationEU. Ó.L-B was granted by ISCIII Support Platforms for
Clinical Research (ISCIII/ PT23/00118/Co-funded by European Union).

\bibliography{main}

\bibliographystyle{abbrvnat}

\end{document}